\documentclass[a4paper,fleqn]{cas-dc}

\usepackage[numbers]{natbib}

\usepackage{amsmath}
\usepackage{amsthm}
\usepackage{tikz}
\usepackage{amssymb}
\usepackage{algorithm}
\usepackage[noend]{algpseudocode}
\usepackage{soul}
\usepackage{nicefrac}
\usepackage{bbm}
\usepackage{color, colortbl}
\usepackage{placeins}

\def\tsc#1{\csdef{#1}{\textsc{\lowercase{#1}}\xspace}}
\tsc{WGM}
\tsc{QE}
\tsc{EP}
\tsc{PMS}
\tsc{BEC}
\tsc{DE}

\algnewcommand\algorithmicinput{\textbf{Input:}}
\algnewcommand\algorithmicoutput{\textbf{Output:}}
\algnewcommand\Input{\item[\algorithmicinput]}
\algnewcommand\Output{\item[\algorithmicoutput]}

\newcommand{\mRefFig}[1]{Fig.~\ref{#1}}

\newcommand{\mRefEq}[1]{Equation~(\ref{#1})} 

\newcommand{\erdos}{Erd\H{o}s-R$\acute{\text{e}}$nyi }
\newcommand{\pp}{p}
\newcommand{\s}{s}

\newcommand{\mG}{G}
\newcommand{\mE}{\mathcal{E}}
\newcommand{\mV}{\mathcal{V}}
\newcommand{\mN}{n}

\newcommand{\mFtr}{\Phi} 
\newcommand{\mNbr}{\mathcal{N}} 
\newcommand{\mdeg}{\mathcal{D}} 
\newcommand{\mDeg}{\mdeg} 
\newcommand{\mDist}{t} 
\newcommand{\mMaxDist}{\lambda} 
\newcommand{\mFtrLen}{\theta} 

\newcommand{\diGa}{\mdeg_{i}^{a}}

\newcommand{\mNbrC}[3]{\mNbr^{\,#1,#2}_{#3}}
\newcommand{\mNormC}[2]{\vert\vert \mFtr^{a}_{#1} - \mFtr^{b}_{#2}\vert\vert_{2}}

\newcommand{\mSim}{X} 

\newtheorem{proposition}{Proposition}

\begin{document}
\let\WriteBookmarks\relax
\def\floatpagepagefraction{1}
\def\textpagefraction{.001}
\shorttitle{Seedless Graph Matching via Tail of Degree Distribution}
\shortauthors{Bozorg et~al.}

\title [mode = title]{Seedless Graph Matching via Tail of Degree Distribution for Correlated \erdos Graphs}

\author[1]{Mahdi Bozorg}
\author[1]{Saber Salehkaleybar}
\author[1]{Matin Hashemi}

\address[1]{Department of Electrical Engineering, Sharif University of Technology, Tehran, Iran}

\nonumnote{The authors are with the Learning and Intelligent Systems Laboratory, Department of Electrical Engineering, Sharif University of Technology, Tehran, Iran. Webpage: http://lis.ee.sharif.edu, E-mails: mehdi.bozorg@ee.sharif.edu, saleh@sharif.edu (corresponding author), matin@sharif.edu.}

\begin{abstract}
The network alignment (or graph matching) problem refers to recovering the node-to-node correspondence between two correlated networks. In this paper, we propose a network alignment algorithm which works without using a seed set of pre-matched node pairs or any other auxiliary information (e.g., node or edge labels) as an input. The algorithm assigns structurally innovative features to nodes based on the tail of empirical degree distribution of their neighbor nodes. Then, it matches the nodes according to these features. We evaluate the performance of proposed algorithm on both synthetic and real networks. For synthetic networks, we generate \erdos graphs in the regions of $\Theta(\log(n)/n)$ and $\Theta(\log^{2}(n)/n)$, where a previous work theoretically showed that recovering is feasible in sparse \erdos graphs if and only if the probability of having an edge between a pair of nodes in one of the graphs and also between the corresponding nodes in the other graph is in the order of $\Omega(\log(n)/n)$, where $n$ is the number of nodes. Experiments on both real and synthetic networks show that it outperforms previous works in terms of probability of correct matching.
\end{abstract}


\begin{highlights}
\item Proposing a network alignment (graph matching) algorithm, which requires neither any seed set of pre-matched nodes, nor any auxiliary node or edge information. 
\item Solving the problem solely based on structural similarities between the two graphs, in specific, based on the tail of empirical degree distribution as node features. 
\item Significantly improving the probability of correct matching compared to previous methods in \erdos graphs.
\end{highlights}

%


\maketitle

\section{Introduction}
\label{sec:introduction}

Graph  matching (or network alignment) between two correlated networks is the problem of finding bijection mapping between the nodes in one network to the nodes in the other network according to structural similarities between them. 
If the two networks have exactly the same structure, the problem reduces to the graph isomorphism problem, but in general, the two networks are only similar, which makes the problem more challenging.  

Network alignment arises in various applications in different fields including computer vision \cite{cho2012progressive},  pattern recognition \cite{conte2004thirty}, autonomous driving \cite{slam_2019}, computational biology \cite{elmsallati2016global,seah2014dualaligner}, and social networks \cite{narayanan2009anonymizing}. 
For instance, in computational biology, protein-protein interactions (PPI) can be modeled as networks. PPI networks of different species can be aligned by solving the network alignment problem which can be useful in investigating evolutionary conserved pathways or reconstructing phylogenetic trees \cite{kuchaiev2010topological}. 

Network alignment algorithms can be classified from different aspects like seed-based algorithms, and seedless algorithms. Seed-based network alignment algorithms work based on a set of pre-matched nodes from the two networks, called seeds \cite{kazemi2016network,mossel2019seeded,zhang2018efficient}, while seedless algorithms do not require any seed set as input \cite{conte2004thirty}. 
Moreover, in order to assist the matching procedure, some algorithms employ node or edge features as a side information (e.g., user names or locations in de-anonymization of social networks \cite{malhotra2012studying, nunes2012resolving}), while some other matching algorithms do not require such prior knowledge and only utilize the structural similarities between the two networks as the most important feature in solving the problem \cite{henderson2011s}. 
In this paper, we propose a \textbf{seedless} network alignment algorithm which \textbf{does not require} either any input seed set, or any input features for the nodes or edges as side information. In other words, the proposed algorithm works solely based on structural similarities between the two correlated networks. 

%
Most of the seed-based network alignment algorithms rely on the idea of percolation, in which the algorithm starts from a small set of pre-matched nodes (seeds), and gradually expands the set of matched nodes by applying some rules on the neighbor nodes of previously matched nodes. 
The pioneering method in this category, which succeeded in de-anonymizing a social network with millions of nodes, was introduced by Narayanan and Shmatikov \cite{narayanan2009anonymizing}. They empirically observed that the proposed algorithm is very sensitive to the size of the seed set. If the size of seed set is too small, the algorithm could not percolate, but if the size exceeds a threshold, the algorithm could successfully percolate and de-anonymize a large portion of the entire network. 
Yartseva and Grossglauser \cite{yartseva2013performance} later proved that such phenomenon happens in random bigraph models. 
Later, Kazemi et al. \cite{kazemi2015growing} proposed a percolation-based method called NoisySeed algorithm. The main advantage of this algorithm, as the name implies, is that the initial seed set can include some incorrectly matched pairs as well. The required size for the seed set as well as the tolerable number of incorrect matches have been investigated in \cite{kazemi2015growing}.

Compared with the above solutions, the seedless algorithms do not require pre-matched node pairs as an input. 
In the literature, several seedless methods have been proposed based on convex relaxations of network alignment problem. For instance, in \cite{lyzinski2015graph}, alignment problem is relaxed as a \linebreak quadratic programming problem, and then, the solution is projected into zeros and ones in order to recover the mapping between nodes of two networks.
%
Some other seedless algorithms rely on computing graph edit distance between the two networks, which is basically the minimum number of edge deletions or insertions required to convert one of the networks to the other one \cite{conte2004thirty, fernandez2001graph}. 
Methods based on convex relaxations or graph edit distance are often much more time consuming than other seedless network alignment algorithms \cite{degreeProfile}. 

Spectral methods are another type of seedless algorithms which align nodes based on eigenvalues and eigenvectors of a transformation of the network's adjacency matrix \cite{carcassoni2002alignment, leordeanu2005spectral}. The main idea in these methods is to obtain Laplacian matrices from adjacency matrices of the two networks and then compute the eigenvectors and eigenvalues of these Laplacian matrices. Next, $k$ number of eigenvectors corresponding to top $k$ eigenvalues are selected to construct a $k$-dimensional feature vector for every node. From these feature vectors, the nodes in two networks can be aligned based on a distance metric.

Besides to the above seedless algorithms, several \linebreak machine-learning based algorithms have been proposed that match nodes based on a set of features which are extracted by processing additional information from nodes, e.g., user-names or locations in social networks \cite{abel2010interweaving, egozi2013probabilistic, nunes2012resolving}. As mentioned before, the proposed method in this paper works \linebreak merely based on structural similarities between the two networks, and does not require any additional features.

Recently, few seedless network alignment algorithms \linebreak have been proposed for \erdos graphs. 
Barak et al. \cite{barak2018nearly} presented a matching algorithm that finds certain small sub-networks that appear in both networks, based on which a set of seeds is formed accordingly. Next, a percolation algorithm extends the selected seeds to match all the nodes. 
This algorithm is designed for \erdos graphs with average node degrees in the range $[n^{o(1)}, n^{1/153}]$ or $[n^{2/3}, n^{1-\epsilon}]$, where $\epsilon$ is a small positive constant. This range covers very sparse or very dense \erdos graphs. Compared with this algorithm, the proposed solution in this paper works on \erdos graphs with average node degrees of order $\log(n)$. In fact, it has been shown that the true graph matching can be recovered with high probability if and only if the average node degree is in the order of $\Omega(\log(n))$ \cite{cullina2016improved}. Thus, our proposed algorithm can work for the minimum value of average node degree that is possible to find the correct matching.

Dai et al. \cite{dai2018performance} proposed another network alignment algorithm for \erdos graphs called canonical labeling. In the first step of this algorithm, the nodes in the two networks are sorted according to their degrees. Then, the top $h$ highest degree nodes in two networks are aligned based on the sorted lists. 
In the second step, each remaining node $j$ gets a binary vector of length $h$. Entry $i$ of this vector is equal to one if node $j$ is connected to $i$-th node in the sorted list. Otherwise, this entry is set to zero. The nodes are then aligned according to these binary feature vectors. Our experiments show that the canonical labeling does not have good performance in the networks with average node degrees of order $\log(n)$ or even $\log^{2}(n)$. 
Ding et al. \cite{degreeProfile} proposed a network alignment algorithm for \erdos graphs with average node degree in three regions including $\Theta(\log^{2}(n))$. In this algorithm, every node is assigned a feature vector containing empirical degree distribution of its neighbors. Then, the minimum distance on these features are used to match the nodes. This algorithm has a relatively higher accuracy in \erdos graphs with average degree of $\log(n)$, but our experiments show that it has lower performance for the graphs with average node degree of order $\log^2(n)$. 

Beside to the mentioned algorithms for \erdos \linebreak graphs, several graph matching algorithms have been proposed with specific applications in PPI networks, social networks, and image databases. Singh et al. \cite{singh2008global} introduced a well-known network alignment algorithm in PPI networks, which is named IsoRank. In this algorithm, the similarity of a node $i$ in one of the network to a node $j$ in the other network depends on how similar are the neighbor nodes of node $i$ to the neighbor nodes of node $j$. More specifically, in the first step of this algorithm, the similarity matrix $R$ is constructed iteratively where entry $R_{ij}$ indicates similarity of node $i$ in one of the network to node $j$ in the other network. In each iteration, entry $R_{ij}$  is computed from other entries in $R$ like $R_{uv}$, where $u$ and $v$ are neighbor nodes of $i$ and $j$ in the two networks, respectively. In the second step, nodes in two networks are aligned according to $R$. 
Later, Zhang et al. \cite{zhang2016final} proposed a network alignment algorithm, called Final algorithm. The Final algorithm can work on both node and edge attributed networks or simple networks without any auxiliary information. Furthermore, this algorithm uses prior knowledge of pairwise alignment preference $H$ matrix, where each entry in this matrix shows likelihood of aligning two corresponding nodes from two input networks. If this prior knowledge is not given, all entries of $H$ are set to the same value, i.e., a uniform distribution. This algorithm iteratively minimizes an objective function, which is constructed from network structure (i.e, adjacency matrix) and nodes and edges attributes.
Zhang et al. \cite{zhang2019multilevel} proposed another network algorithm called Moana. This algorithm aligns nodes in three steps. First, it coarsens the input networks to a structured representation. Next, it aligns the coarsened representation. Finally, the alignment at multi levels is obtained including node level by interpolation. 

Recently, several works \cite{Fey2020Deep,Yu2020Learning} with the applications in the fields of computer vision, used graph neural networks in order to obtain node embedding vectors and match the nodes based on them. These works utilized extracted features from images as inputs to the graph neural network to facilitate the process of graph matching.

In this paper, we propose a seedless network alignment algorithm, which works without any auxiliary information.
The proposed algorithm has two main steps: In the first step, for each node $i$ in any of two correlated networks, we construct a feature vector containing degrees of nodes like $j$ having the following two properties: I) Node $j$ should be in the neighborhood of node $i$. II) Its degrees is in the tail of empirical degree distribution of nodes in neighborhood of node $i$. Due to this property of the proposed algorithm, we call it ``Tail Degree Signature (TDS)" network alignment algorithm. In the second step, we compute a distance metrix between any pair of feature vectors to generate the matrix of distances. Then we use a greedy algorithm or the Hungarian algorithm \cite{kuhn1955hungarian}) to align nodes from the constructed distance matrix. We evaluate the performance of TDS algorithm for both synthetic and real networks. For synthetic networks we select \erdos graphs with average degree of order $\log(n)$ and $\log^2(n)$, which are difficult regions for the network alignment problem \cite{cullina2016improved}.
Experiments show that the proposed TDS algorithm outperforms other related works in both real-world networks and synthetic \erdos graphs with average node degree of order $\log(n)$ and also $\log^2(n)$.

\begin{figure*}
	\begin{center}
		\renewcommand{\arraystretch}{1.4}
		
		\centerline{\includegraphics[width=.9\textwidth]{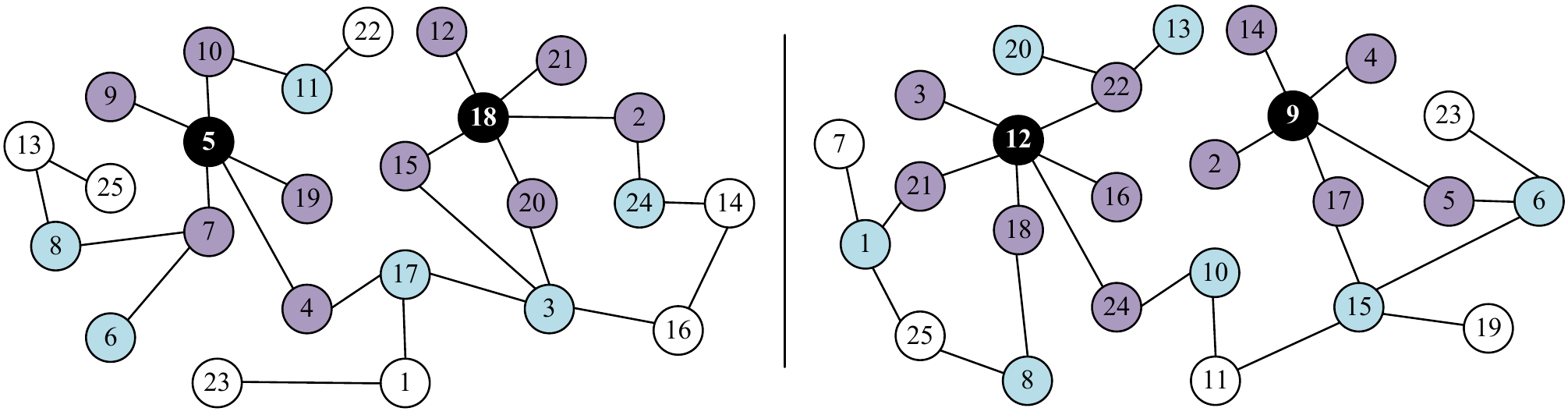}}
		(a)
		\vskip 4mm
		
		\definecolor{mPurple}{RGB}{106,84,133}
		\definecolor{mBlue}{RGB}{44,118,138}
		\centering
		\begin{tabular}{|l|l|rr}
			\cline{1-2}
			\color{mPurple}Distance $t=1$ from $i=18$ & \color{mBlue}Distance $t=2$ from $i=18$    &                             & Feature Vectors:\\
			\cline{1-2}
			$~~~~~~\mNbrC{a}{1}{18}=\{15,12,21,~2,20\}$ & $~~~~~~\mNbrC{a}{2}{18}=\{3,24\}$        &                             & $\mFtr^{b}_{9} =  [1,2,3,4]$\\
			$\mDeg(\mNbrC{a}{1}{18})=\{~2,~~1,~~1,~~2,~2\}$ & $\mDeg(\mNbrC{a}{2}{18})=\{4,~2\}$   &                             & $\mFtr^{b}_{12} = [1,3,1,3]$\\
			$~~~~~~~\,\mFtr^{a,1}_{18}=[1,2]$& $~~~~~~~\,\mFtr^{a,2}_{18}=[2,4]$                   &                             & $\mFtr^{a}_{5} =  [1,3,1,3]$\\
			\cline{1-2}
			\multicolumn{2}{|c|}{$\mFtr^{a}_{18}=\mFtr^{a,1}_{18} | \mFtr^{a,2}_{18} = [1,2,2,4]$} &$\xrightarrow{\hspace*{1cm}}$& $\mFtr^{a}_{18} = [1,2,2,4]$\\
			\cline{1-2}
		\end{tabular}
		\\
		\vskip 4mm
		(b)
		
		\begin{tabular}{c|c|c|c|c|c|}
		\multicolumn{2}{r}{$\cdots$}&\multicolumn{3}{c}{$j=9~~~~~~~~~~~~~~~\cdots~~~~~~~~~~~~~~~j=12$}&\multicolumn{1}{c}{$\cdots$}\\		\cline{2-6}
		$\cdots$ & \multicolumn{5}{c|}{~} \\ 																					\cline{3-3} \cline{5-5}			
		$i=5$    & $\cdots$          & $\mNormC{5}{9}=2.45$  & $\cdots$ & $\mNormC{5}{12}=0.00$  & $\cdots$  \\					\cline{3-3} \cline{5-5}
		$\cdots$ & \multicolumn{5}{c|}{~} \\																					\cline{3-3} \cline{5-5}
		$i=18$   & $\cdots$          & $\mNormC{18}{9}=1.00$ & $\cdots$ & $\mNormC{18}{12}=1.73$ & $\cdots$  \\					\cline{3-3} \cline{5-5}
		$\cdots$ & \multicolumn{5}{c|}{~} \\ 																								\cline{2-6}
		\end{tabular}
		\\
		\vskip 4mm
		(c)
		
	\end{center}
	\caption{
		Feature vectors $\mFtr$ for a selected subset of the nodes, in specific, for nodes $18$ and $5$ in graph $\mG_a$ and nodes $12$ and $9$ in graph $\mG_b$. 
		In this simple example, $\mN=25$, $\mMaxDist=2$, and $\mFtrLen=1$. 
		(a) Neighbors at distance $1$ and $2$ are marked in purple and blue colors, respectively. 
		(b) The computed feature vectors $\mFtr$ for the selected nodes. Computing $\mFtr^a_{18}$ is shown in details. 
		(c) Part of the similarity matrix. 
	}
	\label{fig:methodParts}
\end{figure*}

\section{Problem Definition}
\label{sec:problem}

Network alignment is problem of identifying a bijection mapping between nodes in two structurally similar graphs. Let $\mG_a(\mV_a,\mE_a)$ and $\mG_b(\mV_b,\mE_b)$ be two graphs with node sets $\mV_a$ and $\mV_b$ of size $\mN$, and edge sets $\mE_a$ and $\mE_b$. We denote the edge between nodes $i$ and $j$ by $(i,j)$. 
Let mapping function $\pi:\mV_a \rightarrow \mV_b$ denote a one-to-one mapping between nodes of $\mG_a$ and $\mG_b$. 
The goal in the graph matching problem is to select a matching $\hat{\pi}$ from  $\mN!$ different possible mapping functions in the symmetric group $S_n$ such that:

\begin{equation}
\hat{\pi} = \underset{\pi\in S_n}{\operatorname{argmin}} \; \|A(\mG_b)-P_{\pi}^{T}A(\mG_a)P_{\pi}\|^2_F,
\label{eq:problem_def}
\end{equation}

\noindent where $\|.\|_F$ is Frobenius norm and $A(\mG_a)$ and $A(\mG_b)$ are the adjacency matrices for $\mG_a$ and $\mG_b$, respectively. Moreover, the matrix $P_{\pi}^{T}A(\mG_a)P_{\pi}$ is a simultaneous row/column  permuted version of $A(\mG_a)$, and $P_{\pi}$ is the permutation matrix corresponding to mapping $\pi$ which is defined as:  

\begin{equation}
P_{\pi}[i,j] = \left\{
\begin{array}{lcl}
1 & : & i\in \mV_a,~ j\in \mV_b,~ j = \pi(i),\\
0 & : & \text{otherwise.}
\end{array}
\right.
\end{equation}

In other words, the objective function in \mRefEq{eq:problem_def} measures the number of mis-matched edges between relabeled version of graph $G_a$ based on mapping $\pi$ and graph $G_b$. In the worst case, solving the above optimization problem is NP-hard \cite{rendl1994quadratic}.

For synthetic graphs, we assume that $\mG_a$ and $\mG_b$ are two correlated  \erdos graphs where the original graph \linebreak $\mG(\mV,\mE)$ is generated with parameter $\pp$, i.e, there is an edge between any two nodes with probability $\pp$. Then, two correlated graphs $\mG_a$ and $\mG_b$ are constructed where edge sets $\mE_a$ and $\mE_b$ are sampled from $\mE$ with probability $\s$. In other words, every edge in edge set $\mE$ is in $\mE_a$ and $\mE_b$ with probability $\s$, independently. The vertex set $\mV_a$ is the same as $\mV$, but $\mV_b$ is a permuted version of $\mV$ according to mapping $\pi^{*}$. The matching algorithm tries to recover $\pi^{*}$ given only $\mG_a$ and $\mG_b$. For correlated \erdos graphs, it can be shown \cite{kazemi2016network} that maximum a-posteriori (MAP) estimation is equivalent to minimizing the objective function in \mRefEq{eq:problem_def}. Furthermore, MAP estimator finds the ground truth matching, i.e., $\hat{\pi}=\pi^*$ with high probability if and only if $ps^2=\Omega(\log(n)/n)$ \cite{cullina2016improved}. Hence,  no matching algorithm can return the correct output for values less than $\Omega(\log(n)/n)$.

\section{Tail Degree Signature (TDS) Algorithm}
\label{sec:alg}

Our proposed graph matching algorithm consists of two steps: I) For every node in both graphs, a feature vector is extracted. II) Based on these feature vectors, the nodes in the two subsets are matched.

\subsection{Feature Extraction}
\label{sec:alg:feature}

\noindent \textbf{Method:} 
For every node $i\in \mV_a$, we extract a feature vector $\mFtr^a_i$ based on its neighbor nodes in $\mG_a$ as follows: 
Let $\mNbrC{a}{\mDist}{i}$ be the set of nodes in graph $\mG_a$ whose distance from node $i$ is exactly equal to $\mDist$, where $\mDist \in [1,\mMaxDist]$, and $\mMaxDist$ is the maximum distance that is considered in the feature extraction procedure. For every node $i\in \mV_a$ and every $\mDist \in [1,\mMaxDist]$, set $\mDeg(\mNbrC{a}{\mDist}{i})$ is formed as the  degrees of the nodes in $\mNbrC{a}{\mDist}{i}$, i.e.,

\begin{equation}
\mDeg(\mNbrC{a}{\mDist}{i}) = \{ 
\mdeg^{a}_{i'} ~\vert~ i' \in \mNbrC{a}{\mDist}{i} 
\}.
\end{equation}

Next, for a given integer parameter $\theta$, we pick $\mFtrLen$ of the smallest and $\mFtrLen$ of the largest elements in $\mDeg(\mNbrC{a}{\mDist}{i})$ and put them in feature vector $\mFtr^{a,\mDist}_i$ of size $2 \, \mFtrLen$. 
Finally, feature vector $\mFtr^a_i$ is formed by concatenating vectors $\mFtr^{a,\mDist}_i$ as follows:

\begin{equation}
\mFtr^a_i = \mFtr^{a,1}_i ~|~ \mFtr^{a,2}_i ~|~ \cdots ~|~ \mFtr^{a,\mMaxDist}_i.
\end{equation}

Thus, $\mFtr^a_i$ is a vector of size $2 \, \mFtrLen  \mMaxDist$. By a similar procedure, for every node $j\in \mV_b$, feature vector $\mFtr^b_j$ is also formed. 
\mRefFig{fig:methodParts}(a) shows two example graphs $\mG_a$ and $\mG_b$. 
\mRefFig{fig:methodParts}(b) shows the construction procedure of $\mFtr^a_{18}$ where $\mNbrC{a}{t}{18}$ and $\mDeg(\mNbrC{a}{t}{18})$ are generated according to \mRefFig{fig:methodParts}(a). As few other examples, $\mFtr^a_{5}$, $\mFtr^b_{12}$ and $\mFtr^b_{9}$ are also shown in \mRefFig{fig:methodParts}(b).

\vskip 2mm
\noindent \textbf{Rationale:} 
In constructing the vector $\mFtr^a_i$, we select the degree of nodes in $\mNbrC{a}{\mDist}{i}$ which are in the \textbf{tail} region of empirical degree distribution of nodes in $\mNbrC{a}{\mDist}{i}$. Herein, we give an intuition why such selection is more preferable than considering nodes' degrees outside of this region for $t=1$.  For the original graph $\mG(\mV,\mE)$, i.e, there is an edge between any two nodes in $\mG(\mV,\mE)$ with probability $\pp$. Then $\mG_a$ and $\mG_b$ are constructed where edge sets $\mE_a$ and $\mE_b$ are sampled from $\mE$ with probability $\s$. Thus, $\mG_a$ and $\mG_b$ are two \erdos graphs with $\s\pp$ probability. It can be seen that the degree distribution of node $i$ in graph $\mG_a$ or $\mG_b$ is approximately a normal distribution $N(\mu,\sigma^2)$ with parameters $\mu=(n-1)ps$ and $\sigma=\sqrt{(n-1)(1-ps)ps}$. Let $U_i^a$ be the normalized degree of node $i$ in graph $\mG_a$, i.e., $U_i^a=(\diGa-\mu)/\sigma$. $U_i^b$ is defined similarly in graph $\mG_b$. 

\begin{proposition}
If node $j\in\mV_b$ is the corresponding node of a node $i\in \mV_a$, i.e., $j=\pi^*(i)$, then $U_i^a$ and $U_j^b$ are two correlated random variables with the  correlation coefficient:  $\rho=s(1-p)/(1-ps)$. Otherwise, they are approximately uncorrelated for large $n$.
\label{Prop:2}
\end{proposition}

\begin{proof}
To prove the above statement, it is just needed to compute the following term for the two correlated random variables $U_i^a$ and $U_{\pi^*(i)}^b$:
\begin{equation}
\begin{split}
\mathbb{E}\Big[\diGa &\mathcal{D}^b_{\pi^*(i)}\Big]\\
&=\mathbb{E}\Bigg[\sum_{k\neq i} \mathbbm{1}[(i,k)\in \mE_a]  \sum_{k'\neq \pi^*(i)}\mathbbm{1}[(\pi^*(i),k')\in \mE_b]\Bigg]\\
&\stackrel{(a)}{=}\big((n-1)^2-(n-1)\big)(ps)^2\\
&\qquad+\sum_{k\neq i}\mathbb{E}\Big[\mathbbm{1}[(i,k)\in \mE_a] \mathbbm{1}[(\pi^*(i),\pi^*(k))\in \mE_b]\Big]\\
&\stackrel{(b)}{=}\big((n-1)^2-(n-1)\big)(ps)^2 +(n-1)ps^2,
\end{split}
\end{equation} 
where $\mathbbm{1}[.]$ is the indicator function.\\
(a) Due to the fact that the events ${1}[(i,k)\in \mE_a]$ and \linebreak $\mathbbm{1}[(\pi^*(i),k')\in \mE_b]$ are independent for $k'\neq\pi^*(k)$.\\
(b) The probability of existing an edge between nodes $i$ and $k$ in the original graph $\mG$ is equal to $p$. Moreover, the probability of having that edge in both graphs $\mG_a$ and $\mG_b$ is $s^2$. Hence, the expectation of event in the second sum would be $ps^2$.

Thus, the correlation coefficient between $\diGa$ and $\mathcal{D}^b_{\pi^*(i)}$ would be:
\begin{equation}
\begin{split}
\rho=\frac{\mathbb{E}\Big[\diGa \mathcal{D}^b_{\pi^*(i)}\Big]-\mu^2}{\sigma^2}=s(1-p)/(1-ps).
\end{split}
\end{equation}	
Similarly, for the case of $j\neq \pi^*(i)$, it can be shown that $\rho=s(1-p)/((n-1)(1-ps))$. Hence, the two random variables are approximately uncorrelated for large $n$ if $j\neq \pi^*(i)$.
\end{proof}

Based on the above observation, we can model the two random variables $U_i^a$ and  $U^b_{\pi^*(i)}$ as  $U^b_{\pi^*(i)}=\rho U_i^a + \sqrt{1-\rho^2} Z$ where $Z$ is an independent standard normal variable. To show the advantage of selecting nodes' degrees in the tail of degree distribution, we define the following two metrics between any two nodes $i\in \mV_a$ and $j\in \mV_b$: 
\begin{equation}
\begin{split}
\Delta^{i,j}_{tail}&= \frac{1}{2}\int^{-\frac{1}{2}}_{-\infty} |\hat{p}_{U_i^a}(x)-\hat{p}_{U_j^b}(x)|dx\\
&+ \frac{1}{2}\int_{\frac{1}{2}}^{\infty} |\hat{p}_{U_i^a}(x)-\hat{p}_{U_j^b}(x)|dx\\
\Delta^{i,j}_{center}&=\frac{1}{2}\int^{\frac{1}{2}}_{-\frac{1}{2}} |\hat{p}_{U_i^a}(x)-\hat{p}_{U_j^b}(x)|dx,
\end{split}
\end{equation} 
where $\hat{p}_{U_i^a}(x)$ and $\hat{p}_{U_j^b}(x)$ are the empirical distribution obtained from observing samples of $U_i^a$ and $U_j^b$, respectively. In fact, $\Delta^{i,j}_{tail}$ and $\Delta^{i,j}_{center}$ represent the total variation distances \cite{levin2017markov} of $\hat{p}_{U_i^a}(x)$ and $\hat{p}_{U_j^b}(x)$  in the tail and central domains of distributions, respectively. For a given node $i$, we are interested in comparing $\Delta_{tail}^{i,\pi^*(i)}$ with $\Delta_{tail}^{i,j}$ for any $j\neq \pi^*(i)$. We define the score  $s^{i,j}_{tail}=\Delta_{tail}^{i,\pi^*(i)}/\Delta_{tail}^{i,j}$ for any $j\neq \pi^*(i)$ and $j\in \mV_b$. We expect to have better matching results for higher $s^{i,j}_{tail}$. The score $s^{i,j}_{center}$ is defined similarly. We compare the average of $s^{i,j}_{tail}$ and $s^{i,j}_{center}$ experimentally by generating $100$ samples of $U_i^a$ and $U_j^b$ for two correlation coefficients $\rho=s(1-p)/(1-ps)$ and $\rho=0$. From these samples, one instance of $s^{i,j}_{tail}$ and $s^{i,j}_{center}$ can be computed. \mRefFig{fig:tail_score} shows the average of $s^{i,j}_{tail}$ and $s^{i,j}_{center}$ over $100$ instances against parameter $s$ for $n=1000$ and $p=\log(n)/n$. As can be seen, the score of tail region is about $40\%$ greater than the one for the central region. This observation illustrates that the empirical degree distribution in the tail region is much more robust to sampling parameter $s$.

\begin{figure}[tp]
	\begin{center}
		\includegraphics[width = 0.95\columnwidth]{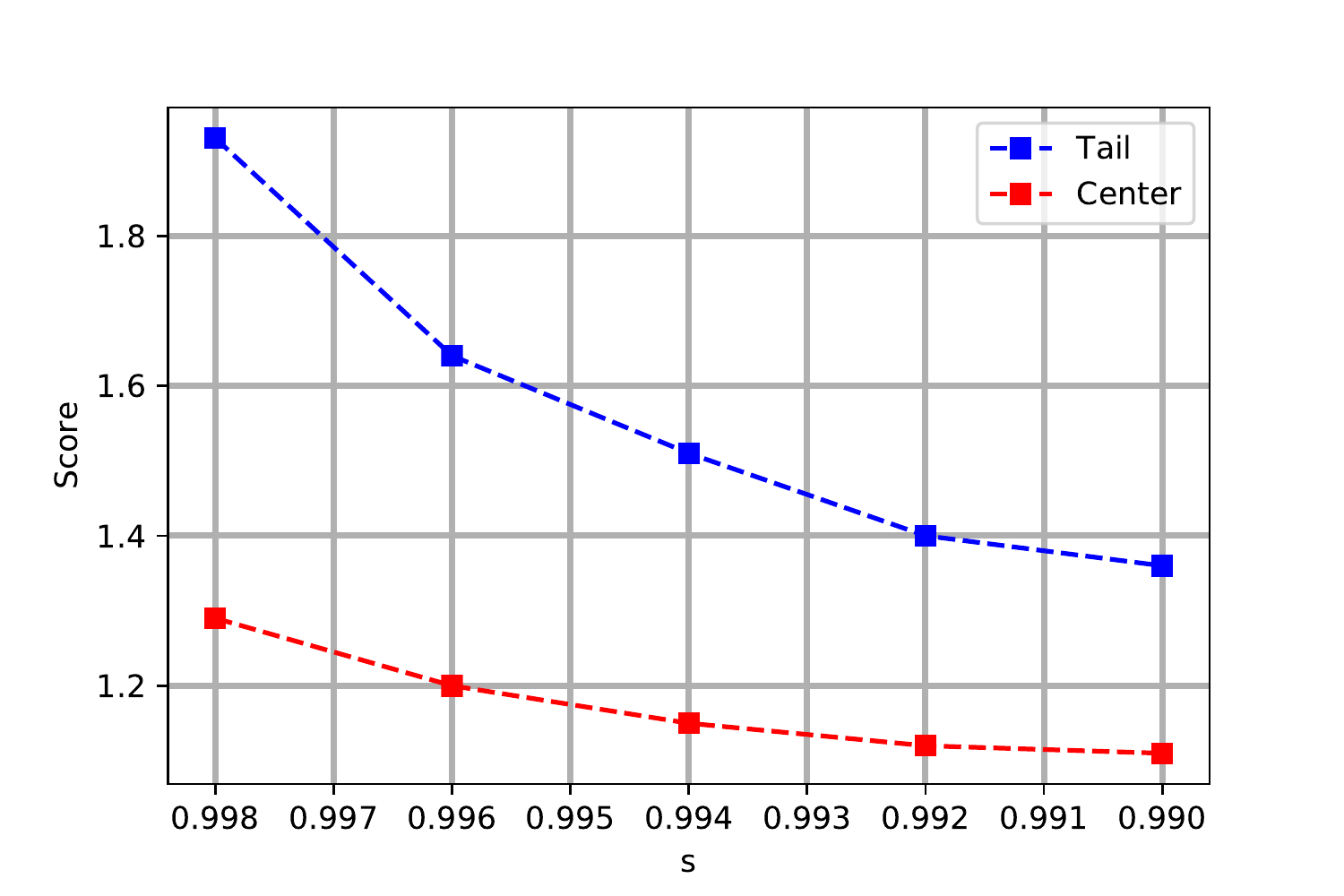}
	\end{center}
	\vskip -4mm
	\caption{The average of $s^{i,j}_{tail}$ and $s^{i,j}_{center}$ over $100$ instances against parameter $s$ for $n=1000$ and $p=\log(n)/n$. 
	}	
	\vskip -1mm
	\label{fig:tail_score}
\end{figure}

\vskip 2mm
\noindent \textbf{Complexity Analysis:} 
For every node $i\in \mV_a$, we run BFS algorithm with root node $i$ and obtain all nodes in $\mNbrC{a}{\mDist}{i}$ for every $t\in [1,\lambda]$. Since the average degree of each node is in the order of $O((n-1)ps)$, the average number of neighbor nodes up to distance $\lambda$ is in the order of $O((nps)^{\lambda})$. Thus, the time complexity of this part is $O((nps)^{2\lambda})$. Moreover, it takes $O(t(nps)^t\log(nps))$ to sort the nodes in $\mNbrC{a}{\mDist}{i}$ and construct $\mFtr^{a,t}_i$. Therefore, the total time complexity of the feature extraction step for all $\mN$ nodes is in the order of $O \Big( \big( (nps)^{2\lambda}+\lambda (nps)^{\lambda} \log(nps) \big) \times \mN \Big)$. For $p=\log(n)/n$ and $\lambda=2$, the time complexity is simplified to $O(n\log^4(n))$. For $p=\log^2(n)/n$, it is in the order of $O(n\log^8(n))$.

\subsection{Matching Method}
\label{sec:alg:seed}

First, we compute similarity matrix (or distance matrix) $\mSim$ between $\mV_a$ and $\mV_b$. In particular, element $(i,j)$ in this matrix is equal to: $\mSim_{ij} = \mNormC{i}{j}$ where $i \in \mV_a$ and $j \in \mV_b$.

Next, we form the set of matched pairs between $\mV_a$ and $\mV_b$, by executing Hungarian algorithm on the similarity matrix $X$. More specifically, Hungarian algorithm selects $\mN$ number of entries from matrix $X$, where from each column and each row, exactly one entry is chosen and the selected entries minimize the following cost:
\begin{equation}
cost=\frac{1}{\mN}\sum X_{ij}.
\label{eq:hungarian_cost}
\end{equation}

In other words, by running Hungarian algorithm, we form a mapping $\pi$ between $\mV_a$ and $\mV_b$ that has minimum mean of similarity distance over all possible choices. We call this version of the proposed method as ``TDS-h algorithm''.

Another option, instead of using Hungarian algorithm, is to use the following simple greedy algorithm. We select the minimun element $X_{ij}$ in matrix $X$ and align node $i \in \mG_a$ with node $j \in \mG_b$, and delete row $i$ and column $j$ from matrix $X$. This process is repeated $n$ times.  We call this version of the proposed method as ``TDS-g algorithm''.

As an example, \mRefFig{fig:methodParts}(c) shows $l_2$-norm distances between constructed feature vectors $\mFtr^{a}_{i}$ and $\mFtr^{b}_{j}$ from \mRefFig{fig:methodParts}(b). Four elements of the similarity matrix are shown in the figure. Either of the two matching methods can be applied. In this example, nodes $5$ and $18$ in graph $\mG_a$ are matched to nodes $12$ and $9$ in graph $\mG_b$, respectively.

\vskip 2mm
\noindent \textbf{Complexity Analysis:} 
Time complexity of finding similarity matrix $X$ is in the order of $O(\mN^{2})$. Both matching methods are in the order of $O(\mN^{3})$.

\section{Experimental Evaluation}
\label{sec:exp}

The proposed seedless graph matching algorithm, called tail degree signature (TDS), is experimentally evaluated in this section.  
The constant parameters $\mMaxDist$ and $\mFtrLen$ are set to $2$ and $10$, respectively. 
The algorithm is implemented in Python language.

\begin{figure*}
	\centering
		\includegraphics[width=.95\textwidth]{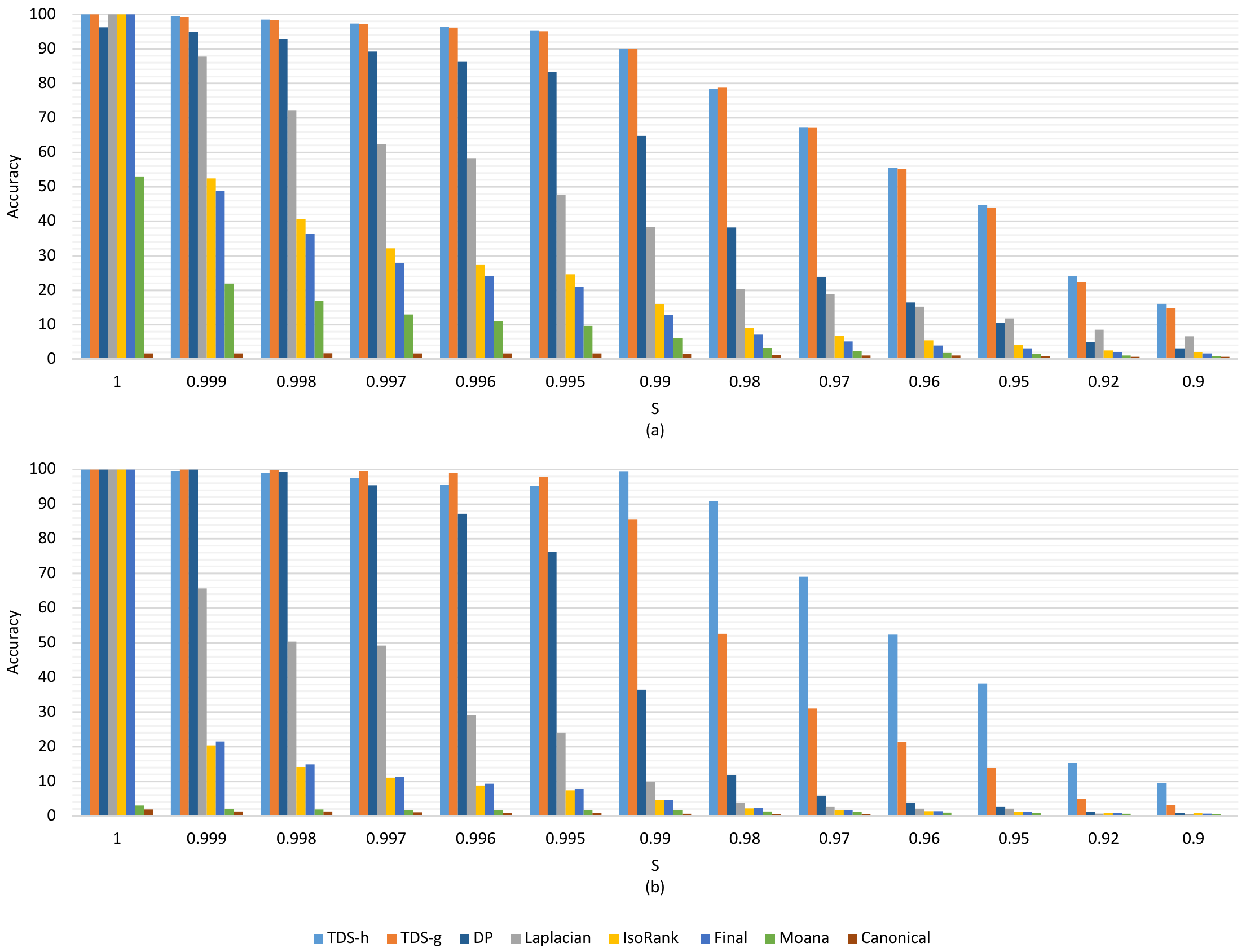}
		\caption{Accuracy of the proposed TDS algorithm and previous seedless algorithms, versus different values of $s$, for \erdos graphs with (a) $\pp = \log(\mN)/\mN$ and (b) $\pp = \log^2(\mN)/\mN$, and $\mN=1000$. The lower the value of $s$ the more difficult it becomes to match the graphs.} 
		\vskip -2mm
	\label{figaccerdos}
\end{figure*}

\subsection{Accuracy}

The two versions of TDS algorithm (TDS-h and TDS-g) are compared with recent seedless graph matching algorithms on \erdos graphs with $p=\log(n)/n$ and $p=\log^2(n)/n$. 
In particular, we consider Degree Profile (DP) \cite{degreeProfile}, Laplacian \cite{carcassoni2002alignment, leordeanu2005spectral}, Canonical labeling \cite{dai2018performance}, IsoRank \cite{singh2008global}, Final \cite{zhang2016final}, and Moana \cite{zhang2019multilevel} algorithms.

\mRefFig{figaccerdos}(a) shows accuracy of TDS algorithms and the other methods versus $s$. Every value in this figure shows the average accuracy for $50$ randomly generated \erdos graphs with $\pp = \log(\mN)/\mN$ and $\mN=1000$. 
As shown in the figure, both versions of TDS algorithm achieve much higher accuracy compared to the other algorithms. Moreover, they yield accurate solutions for lower values of $s$. 
For instance, at $s=0.98$, TDS-h and TDS-g achieve about $80\%$ accuracy, while DP and Laplacian reach about $40\%$ and $20\%$ accuracy, respectively, and the accuracy of all the other methods are less than $10\%$. 
At $s=0.95$, TDS-h and TDS-g achieve about $45\%$ accuracy, while the accuracy of all the other methods are about or less than $10\%$.

\mRefFig{figaccerdos}(b) presents the same comparisons as above for \erdos graphs with $\pp = \log^{2}(\mN)/\mN$. In this region, TDS-h has higher accuracy than TDS-g for lower values of $s$. Both versions of TDS algorithm achieve much higher accuracy compared to the other algorithms. 
For instance, at $s=0.98$, TDS-h and TDS-g achieve about $90\%$ and $50\%$ accuracy, respectively, while all the other methods are about or less than $10\%$. 

\begin{figure}[tp]
	\centering
	\includegraphics[width=.82\columnwidth]{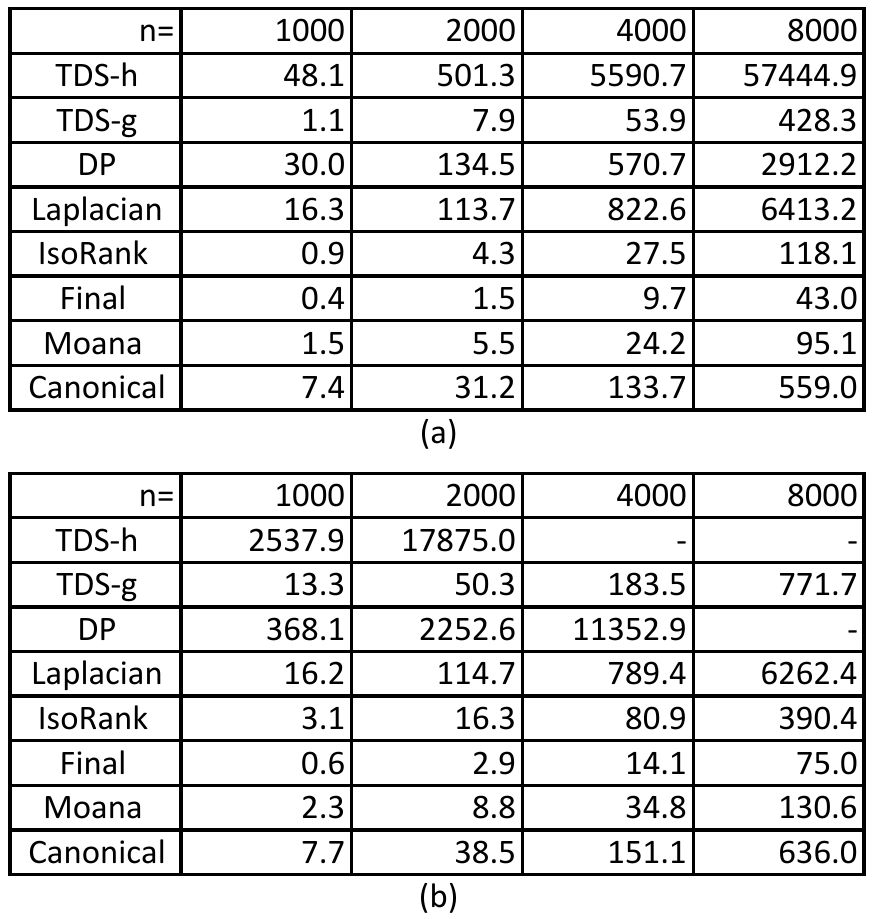} 
	\vskip -1mm
	\caption{Runtime of the proposed TDS algorithm and previous seedless algorithms (in seconds), versus different values of $\mN$, for \erdos graphs with (a) $\pp = \log(\mN)/\mN$ and (b) $\pp = \log^2(\mN)/\mN$, and $s=0.99$. }	
	\label{fig:runtime:erdos}
\end{figure}

\subsection{Runtime}

\mRefFig{fig:runtime:erdos} compares runtimes of the considered algorithms on \erdos graphs with $\pp = \log(\mN)/\mN$ and $\pp = \log^2(\mN)/\mN$ for $n=\{1000, 2000, 4000, 8000\}$ and $s=0.99$. An empty value in \mRefFig{fig:runtime:erdos} denotes that we stopped (killed) the process because the runtime exceeded 16 hours. 

As can be seen, for $\pp = \log^2(\mN)/\mN$, IsoRank, Final, and Moana have smaller runtimes compared to the other methods. TDS-g, Laplacian, and Canonical have comparable runtimes, while TDS-g have much higher accuracy (\mRefFig{figaccerdos}(b)). The runtimes of TDS-h and DP grow dramatically as the number of nodes increases.

\begin{figure*}
	\centering
	\includegraphics[width=0.6\textwidth]{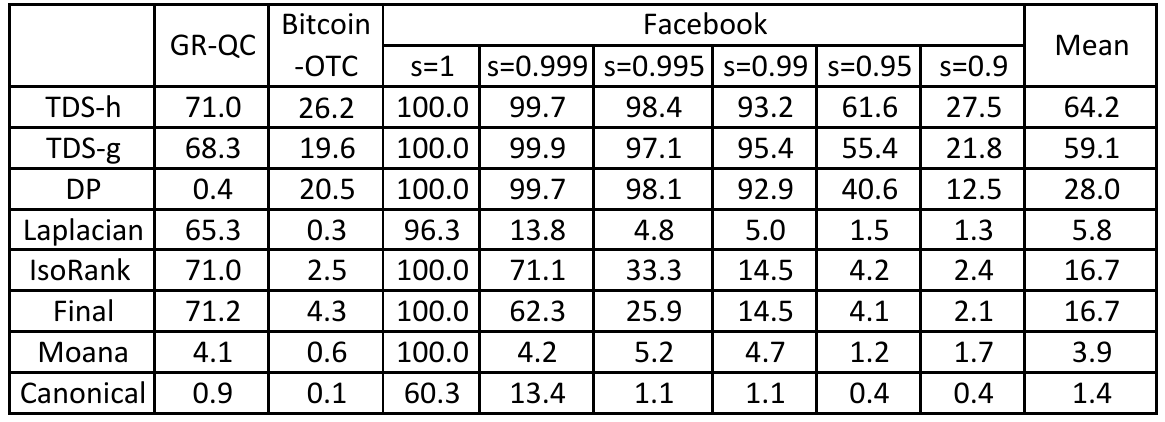} 	
	\vskip -1mm
	\caption{Comparing accuracy of TDS with previous seedless algorithms, in real-world networks. The last column shows the mean accuracy, i.e., the geometric mean of all the other columns.}	
	\label{fig:acc:real}
\end{figure*}

\subsection{Real-world Networks} 

In addition to \erdos graphs, we also evaluated the proposed TDS algorithm on the following three real-world networks and compared it with previous seedless algorithms.  

\begin{itemize}

	\item Bitcoin-OTC \cite{data-bitcoin1, data-bitcoin2}: This is who-trusts-whom network of people who traded using Bitcoin cryptocurrency on a platform called Bitcoin-OTC. It contains 5,881 nodes and 35,592 edges. Members (nodes) on this platform can rate other members (nodes) in the range -10 to 10. To use Bitcoin-OTC as a benchmark for evaluating graph matching algorithms, we consider the following two unweighted networks. The first network contains all the nodes and edges in the original graph, and the second network contains only the positive edges.

	\item GR-QC \cite{data-GR-QC}: arXiv GR-QC (General Relativity and Quantum Cosmology) collaboration network contains 5241 nodes and 11,923 edges. Each author is represented by a node. Two nodes are connected if the authors have at least one common paper on GR-QC category from January 1993 to April 2003. This dataset contains two networks which are permuted version of one another, i.e., $s=1$ \cite{data-GR-QC}.

	\item Facebook \cite{data-facebook}: This dataset contains ``friends lists" from Facebook, which was collected from a survey using a Facebook app. The dataset includes 4039 nodes and 88234 edges. For the task of network alignment, we added some noises to the Facebook network edges, i.e, we constructed a new network with the same set of nodes as Facebook network while each edge in Facebook network is preserved in the new network with probability $s$. 
\end{itemize}

In \mRefFig{fig:acc:real}, we compare the accuracy of TDS with other seedless algorithms. In almost all cases,  the proposed algorithm outperforms the other algorithms. 
For instance, in GR-QC benchmark, the accuracy is around $70\%$ in TDS-h, TDS-g,  Laplacian, IsoRank, and Final algorithms, while the other methods have less than $5\%$ accuracy. 
In Bitcoin-OTC and Facebook benchmarks, TDS-h, TDS-g, and DP have much higher accuracy compared to the other methods. 

The last column in \mRefFig{fig:acc:real} shows the geometric mean of the other columns. As it can be seen, TDS-h and TDS-g achieve the mean accuracy of about $60\%$, while the mean accuracy of DP is about $30\%$, and the other methods have less than $20\%$ accuracy.

\section{Conclusion}
In this paper, we proposed a seedless graph matching algorithm for correlated \erdos graphs. We introduced node features based on tail of degree distribution. We showed that this approach has advantages with respect to matching nodes based on center of degree distributions. Our experiments showed that the proposed algorithm outperforms other related works for several real networks as well as \erdos graphs with average degree of order $\Theta(\log(n))$ and $\Theta(\log^2(n))$.

\bibliographystyle{cas-model2-names}
\bibliography{references}

\end{document}